\documentclass[12pt]{article}%
\usepackage{amsmath,amssymb,amsthm,amsfonts}
\usepackage{wasysym}
\usepackage{subcaption}
\usepackage{amsmath}
\usepackage{amsfonts}
\usepackage{amssymb}
\usepackage{graphicx}
\usepackage{color}%

\providecommand{\U}[1]{\protect\rule{.1in}{.1in}}
\newtheorem{theorem}{Theorem}

\newtheorem{corollary}[theorem]{Corollary}

\newtheorem{lemma}[theorem]{Lemma}

\begin{document}

\title{RS Flip-Flop Circuit Dynamics Revisited}
\author{Aminur Rahman\\Department of Mathematical Sciences\\New Jersey Institute of Technology\\Newark, NJ 07102-1982\\ar276@njit.edu\\* * * \\Denis Blackmore\\Department of Mathematical Sciences\\New Jersey Institute of Technology\\Newark, NJ 07102-1982\\denis.l.blackmore@njit.edu}
\date{}
\maketitle

\begin{abstract}
Logical RS flip-flop circuits are investigated once again in the context of
discrete planar dynamical systems, but this time starting with simple bilinear
(minimal) component models based on fundamental principles. The dynamics of
the minimal model is described in detail, and shown to exhibit some of the
expected properties, but not the chaotic regimes typically found in
simulations of physical realizations of chaotic RS flip-flop circuits. Any
physical realization of a chaotic logical circuit must necessarily involve
small perturbations - usually with quite large or even nonexisting derivatives
- and possibly some symmetry-breaking. Therefore, perturbed forms of the
minimal model are also analyzed in considerable detail. It is proved that
perturbed minimal models can exhibit chaotic regimes, sometimes associated
with chaotic strange attractors, as well as some of the bifurcation features
present in several more elaborate and less fundamentally grounded dynamical
models that have been investigated in the recent literature. Validation of the
approach developed is provided by some comparisons with (mainly simulated)
dynamical results obtained from more traditional investigations.

\end{abstract}


\medskip


\bigskip

\textbf{Since the advent of electronic computing devices, flip-flops have
played a major role in processing arithmetic operations. The first flip-flop
to be designed was the RS flip-flop (RSFF). It has the ability to store bits,
and hence was used in adders and primitive electronic calculators. The RSFF
was even used in the video game \textquotedblleft pong\textquotedblright. The
only logical flaw of this circuit is the ambiguity of the output when the
inputs are both high voltage. In order to rectify this flaw other flip-flops
were designed. In the 1980s Leon Chua designed the simplest circuit to exhibit
chaos. In more recent years, Chua's circuit and the flaw of the RSFF have been
exploited to design chaotic flip-flops. These chaotic flip-flops have the
potential to be employed in random number generation, encryption, and fault
tolerance. However, in order to exploit their properties they need to be
studied further. Since experiments with large systems become difficult,
tractable mathematical models that are amenable to analysis via the tools of
dynamical systems theory are of particular value. The model needs to be simple
enough to simulate with ease and agree qualitatively with experiments for
small systems. We model the RSFF circuit as a discrete dynamical system -
called the minimal model - focusing on the ambiguous output case. The minimal
model is perturbed in various ways to reproduce qualitative behavior observed
in experiments for the chaotic RSFF, and the minimal approach can be readily
generalized to handle other logical circuits.}

\section{Introduction}

Logical circuits are constructed using logic gates representing propositional
connectives such as AND (conjunction), OR (disjunction) and their negation.
They, or more precisely physical approximations of ideal logical circuits,
have important applications in multiplexers, registers, pseudo-random number
generators, quantum network modeling and, in fact, virtually every
microprocessor (see, e.g.\cite{Bos}). Consequently, having efficient methods
for analyzing and predicting their behavior, such as by continuous or discrete
dynamical systems models is of great value in design, analysis and evaluation.

A well-known example with many applications (\emph{e.g}. in quantum networks
\cite{ztkbn}) is the RS Flip-Flop circuit in Fig.1, which is a feedback
circuit comprising two NOR gates (which are negations of OR gates). Note that
an OR gate has an output of 1 when at least one of the inputs is 1 and an
output of 0 when both inputs are zero, with 1 denoting true and 0 false. An
ideal \emph{RS flip-flop circuit} (\emph{RSFF circuit}) is a logical feedback
circuit represented in Fig. 1, with input/output behavior described in Table
1, which shows the \emph{set} ($S$) and \emph{reset} ($R$) inputs for the
circuit consisting of two \emph{NOR gates} with outputs $Q$ and $Q^{\prime}$.
The input to the output, denoted by $(Q_{n},Q_{n}^{\prime})\rightarrow
(Q_{n+1},Q_{n+1}^{\prime})$ may be regarded as the action of a map from the
plane $\mathbb{R}^{2}:=\{(x,y):x,y\in\mathbb{R}\}$ into itself, where
$\mathbb{R}$ denotes the real numbers and in the ideal case, the coordinates
assume the binary values $\{0,1\}$. This suggests that the behavior of the
RSFF circuit might possibly be effectively modeled by the iterates of a planar
map, which comprise a discrete dynamical system (DDS), and its perturbations
derived directly - albeit in a simplified manner - from a faithful
interpretation of the circuit itself. And this is our focus in the sequel,
which should be contrasted with our prior investigation \cite{BRS}. The model
to be studied in what follows, fundamentally grounded as it is in the basic
principles of ideal logic circuit theory, is substantially more compelling and
useful than the one investigated in \cite{BRS}, which was primarily an
\emph{ad hoc} construct designed to mimic the known behavior of RSFF realizations.

\begin{center}
\begin{figure}[th]
\centering
\includegraphics[width=3.5in]{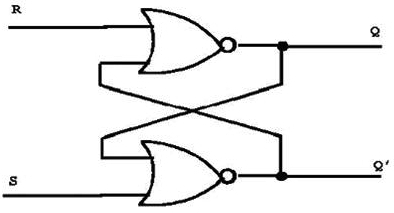}\caption{$R$-$S$ flip-flop
circuit}%
\label{circuit}%
\end{figure}
\end{center}


\noindent The binary input/output behavior, with 0 and 1 representing false
and true, respectively, is given in the following table.

\begin{center}%
\begin{tabular}
[c]{|c|c|c|c|}\hline
$S$ & $R$ & $S_{1}:=Q$ & $R_{1}:=P$\\\hline
1 & 0 & \multicolumn{1}{|c|}{$%
\begin{array}
[c]{c}%
1
\end{array}
$} & $%
\begin{array}
[c]{c}%
0
\end{array}
$\\\hline
0 & 1 & \multicolumn{1}{|c|}{$%
\begin{array}
[c]{c}%
0
\end{array}
$} & $%
\begin{array}
[c]{c}%
1
\end{array}
$\\\hline
1 & 1 & \multicolumn{1}{|c|}{$%
\begin{array}
[c]{c}%
0
\end{array}
$} & $%
\begin{array}
[c]{c}%
0
\end{array}
$\\\hline
0 & 0 & \multicolumn{1}{|c|}{1 or 0} & 0 or 1\\\hline
\end{tabular}

\medskip

Table 1. Binary input/output of $R$-$S$ flip-flop circuit
\end{center}

\medskip

From the DDS perspective, our first goal is to construct the simplest map of
the plane (based on bilinear representations of the NOR gates) that models the
logical properties of the RSFF circuit, with iterates that exhibit most of the
interesting properties that follow from basic analysis or have been observed
in the dynamics of physical realizations and their mathematical models. This
begs the question of how the dynamics of the planar map models are to be
compared with that obtained directly from measurements of physical
realizations of the RSFF circuit and its flip-flop relatives and the analysis
and simulations of the usual mathematical models, which we shall endeavor to
address in what follows.

The usual form of the mathematical models is traceable back to the pioneering
work of Moser \cite{Moser} ; namely, associated three-dimensional systems of
piecewise-smooth, first-order, nonlinear autonomous ordinary differential
equations (ODEs) obtained from applying Kirchhoff's laws to the realizations.
These realizations typically comprise such elements as capacitors, inductors
and nonlinear resistors and exhibit highly oscillatory, very unstable and even
chaotic dynamics (\emph{metastable operation}), as experimentally observed in
such studies as \cite{KA, kac, LMH}, where tunnel diodes of the type used in
Chua's circuit (see \cite{CG,chua,CWHZ,msd} ) are the key ingredients in the
construction of the nonlinear resistors. There are several connections between
the solutions of the model ODEs and iterates of maps that can be used for
dynamical comparisons, among which are the following: as observed by Hamill
\emph{et al. }\cite{hdj}, the autonomous nature of the logical circuit
equations allows dynamical analysis via the iterates (snapshots) of a fixed
time map, which can be reduced to a planar map in special cases as shown in
Kacprzak and Albicki \cite{KA} and Kacprzak \cite{kac}. In addition, the
application of a standard explicit one-step integration method is tantamount
to the iteration of the map describing the scheme, thus enabling the
(approximate) reduction from a continuous to a discrete dynamical system as
shown \emph{e.g}. in Danca \cite{Danca}. Of course, there is the well-known
method of employing Poincar\'{e} sections to analyze three-dimensional
continuous dynamical systems using two-dimensional discrete dynamical systems,
which has been employed in numerous investigations of logical circuit
realizations such as Cafagna and Grassi \cite{CG}, Murali \emph{et al}.
\cite{msd}, Okazaki \emph{et al}. \cite{okt} and Ruzbehani \emph{et al}.
\cite{rzw}.

As it turns out, logical circuits can be realized using the famous circuit of
Chua and its generalizations \cite{chua, CWHZ}, which have depended heavily
powerful tools such as Poincar\'{e} maps, Melnikov functions and normal forms
from the modern theory of dynamical systems and bifurcation theory (see
\cite{deMelo,guho,kathas, Kuz,palmel,Smale,wigbook}) for their analysis. There
is also another interesting connection between realization of logical circuits
and nonlinear maps that might afford an opportunity for comparisons with the
results obtained from our two-dimensional DDS models; namely, the approach in
the work of Ditto \emph{et al}. \cite{DMM}, which features the notion of
reconfigurable logic gates comprising connected NOR and NAND gates constructed
using 1-dimensional DDS and associated thresholds.

Our investigation begins in Section II, where we define a (minimal) planar map
model of an ideal RSFF circuit - derived directly from the logical circuit -
that plays a foundational role throughout the sequel. Moreover, we derive some
basic properties of the minimal map concerning such things as smoothness and
the existence and description of an inverse. This is followed in Section III
with a more thorough analysis of the fixed points of the minimal map -
including a local stability analysis and an analysis of stable and unstable
manifolds. As a result of this more detailed investigation, we find that the
dynamics is quite regular when the domain of the map is appropriately
restricted. Next, in Section IV, after observing that the minimal map is
$C^{1}$-structurally stable, we prove that (1-dimensional) chaos can be
generated by arbitrarily small $C^{0}$ perturbations of the map. In
particular, it is shown that a tent map can be embedded in a (1-dimensional)
stable manifold of a fixed point by such a perturbation, thereby inducing
chaotic tent dynamics. We then prove the existence of several types of more
substantial (2-dimensional) chaotic regimes for arbitrarily small $C^{0}$
perturbations in Section V. In particular, we show that arbitrarily $C^{0}$
small perturbations can be fashioned to produce transverse intersections in
homoclinic orbits and heteroclinic cycles that generate chaos, horseshoe
chaos, multihorseshoe strange chaotic attractors, and
Neimark--Sacker bifurcations. In Section VI, we illustrate our theorems by
making liberal use of numerical simulations of our perturbed models. Finally,
in Section VII, we briefly summarize our results, describe some interesting
areas of application and discuss plans for related future research.

\section{A Minimal Discrete Dynamical System Model}

To obtain the minimal DDS model for the RSFF circuit, we start with the
simplest continuous extension of the input/output map of the NOR gate to the
unit square in the plane, which is clearly given by the function
$\mathcal{N}:I^{2}\,\left(  \subset\mathbb{R}^{2}\right)  \rightarrow
I\,\left(  \subset\mathbb{R}\right)  $ defined as%
\begin{equation}
\mathcal{N}(x,y):=\left(  1-x\right)  \left(  1-y\right)  ,\label{eq1}%
\end{equation}
where $I^{2}:=I\times I:=[0,1]\times\lbrack0,1]\subset\mathbb{R}^{2}$ is the
unit square and $x$ and $y$ are the inputs. Observe that this quadratic
function is bilinear in the variables $(1-x)$ and $(1-y)$, and behaves
precisely like the pure logical NOR gate when $x,y\in\{0,1\}$. Using the
functional representation (\ref{eq1}) for each of the (NOR gate) components of
the RSFF circuit, we obtain the continuous extension for the RSFF,%

\begin{equation}
Q :=\frac{S(1-R)}{R+S-RS},\; Q^{\prime}:=\frac{R(1-S)}{R+S-RS}.
\end{equation}

Since the threshold outputs are fed back in a similar manner to the RSFF
ouputs, we make the naive assumption that the logical behavior is preserved
and we readily obtain its minimal extension to the unit square in the form of
a map $F:I^{2}\rightarrow I^{2}$ defined as%

\begin{equation}
F=\left(  \xi,\eta\right)  :I^{2}\,\left(  \subset\mathbb{R}^{2}\right)
\rightarrow I^{2}\,\left(  \subset\mathbb{R}^{2}\right)  ,\label{eq2}%
\end{equation}
where the coordinate functions are
\begin{equation}
\xi(x,y):=\frac{y(1-x)}{x+y-xy},\;\eta(x,y):=\frac{x(1-y)}{x+y-xy},\label{eq3}%
\end{equation}
with the $x$- $y$-coordinates playing the roles of the previous threshold
ouput, and $\xi$ - $\eta$ coordinates playing the roles of the current threshold ouput.

This map generates a discrete (semi-) dynamical system in terms of its forward
iterates determined by $n$-fold compositions of the map with itself, denoted
as $F^{n}$, where $n\in\mathbb{N}$, the set of natural numbers. We shall
employ the usual notation and definitions for this system; for example, the
\emph{positive semiorbit }of a point $p\in\mathbb{R}^{2}$, which we denote as
$O_{+}(p)$, is simply defined as $O_{+}(p):=\left\{  F^{n}(p):n\in
\mathbb{Z},\,n\geq0\right\}  $, and all other relevant definitions are
standard (\emph{cf}. \cite{guho, kathas, palmel, wigbook}). The minimal map
(\ref{eq2}) is real-analytic $(C^{\omega})$ on $I^{2}$ except at the origin,
where it is not even well-defined. Consequently, we should really consider $F$
to be defined on $I^{2}\smallsetminus\{(0,0)\}$ and actually on $X:=I^{2}%
\smallsetminus\{(0,0),(1,1)\}$ if we wish to avoid the origin for all forward
iterates of $F$.

\subsection{Basic analytical properties of the minimal model}

First we take note of some of the analytical properties of the map
(\ref{eq2}), which are listed in what follows. The proofs of all these results
are straightforward and left to the reader.

\begin{itemize}
\item[(A1)] The map $F\in C^{\omega}\left(  X\right)  $.

\item[(A2)] $F$ is $\mathbb{Z}_{2}$-symmetric in the sense that%
\[
R\circ F=F\circ R
\]
for the reflection $R$ in the line $y=x$.

\item[(A3)] $F\left(  X\right)  =Y:=\left\{  (\xi,\eta)\in\mathbb{R}^{2}%
:0\leq\xi,\eta\text{ and }\xi+\eta<1\right\}  \cup\left\{
(1,0),(0,1)\right\}  $.

\item[(A4)] $F\left(  \{(x,0):0<x\leq1\}\right)  =(0,1)$, $F\left(
\{(0,y):0<y\leq1\}\right)  =(1,0)$, and if $\rho(a):=\{(x,ax):0<x,a\}\cap X$
is the ray through the origin sans the origin in $X$, then $F\left(
\rho(a)\right)  \subset\{(\xi,(a-1)+a\xi):0<\xi\}\cap Y$. Moreover, as
$(x,y)\rightarrow(0,0)$ along the ray $\rho(a)$, $F(x,y)$ converges to the
point of intersection of the line defined by $x+y=1$ with the ray $\rho\left(
1/a\right)  $. This shows just how singular the formulas (3) are at the origin.

\item[(A5)] The derivative (matrix) for $F$ on $X$ is%
\begin{equation}
F^{\prime}\left(  x,y\right)  =\left(  x+y-xy\right)  ^{-2}\left(
\begin{array}
[c]{cc}%
-y & x\left(  1-x\right) \\
y\left(  1-x\right)  & -x
\end{array}
\right)  ,\label{eq4}%
\end{equation}
with determinant%
\begin{equation}
\det F^{\prime}\left(  x,y\right)  =xy\left(  x+y-xy\right)  ^{-3},\label{eq5}%
\end{equation}
which, not surprisingly in view of (A3), shows that the implicit function
theorem cannot guarantee the existence of a local smooth inverse along the $x
$- and $y$-axes.

\item[(A6)] In fact, the inverse of $F$, where it exists, is given as%
\begin{equation}
F^{-1}\left(  \xi,\eta\right)  =\left(  \frac{1-\xi-\eta}{1-\eta},\frac
{1-\xi-\eta}{1-\xi}\right)  ,\label{eq6}%
\end{equation}
which is clearly in $C^{\omega}\left(  Y\smallsetminus\{(1,0),(0,1)\}\right)
$.
\end{itemize}

\section{Dynamics of the Minimal Model}

We shall analyze the deeper dynamical aspects of the model map (2)-(3) for
various parameter ranges in the sequel, but first we dispose of some of the
more elementary properties that follow directly from the definition and
(A1)-(A6), leaving the simple proofs once again to the reader.

\begin{itemize}
\item[(D1)] If we restrict $F$ to $\hat{X}:=\{(x,y):0<x,y$ and $x+y<1\}$, and
denote this restriction by $\hat{F}$, it determines a full dynamical system
defined as%
\[
\left\{  \hat{F}^{n}:n\in\mathbb{Z}\right\}  ,
\]
which, for example, allows the definition of the full \emph{orbit }of a point
$p\in\mathring{Y}$ as
\[
O(p):=\left\{  \hat{F}^{n}(p):n\in\mathbb{Z}\right\}  .
\]

\item[(D2)] The line $y=x$ is $F$-invariant, while the $x$- and $y$-axes are
$F^{2}$-invariant.

\item[(D3)] Both of the points $(1,0)$ and $(0,1)$ are fixed points of $F^{2}
$.
\end{itemize}

\subsection{Analysis of the fixed and periodic points}

The properties of the fixed and periodic points of our model map shall be
delineated in a series of lemmas, which follow directly from the results in
the preceding sections and fundamental dynamical systems theory (as in
\cite{guho, kathas, palmel, wigbook}). Our first result is the following,
which has a simple proof that we leave to the reader.

\begin{lemma}
\label{L1} The only fixed point of $F$ in $X$ is
\[
p_{\ast}=\left(  x_{\ast},y_{\ast}\right)  =\left(  \frac{3-\sqrt{5}}%
{2}\right)  \left(  1,1\right)  \cong0.38197\left(  1,1\right)  ,
\]
which is a saddle point with eigenvalues
\[
\lambda_{s}=-\left(  \frac{3-\sqrt{5}}{2}\right)  ,\;\lambda_{u}=-\left(
\frac{1+\sqrt{5}}{2}\right)  .
\]
This fixed point has linear stable and stable manifolds given as
\[
W_{lin}^{s}\left(  p_{\ast}\right)  =\{(x,x):x\in\mathbb{R}\}\text{ \emph{and}
}W^{s}\left(  p_{\ast}\right)  =\{(x,x):0<x<1\},
\]
and the linear unstable manifold
\[
W_{lin}^{u}\left(  p_{\ast}\right)  =\{(x,-x+2x_{\ast}):x\in\mathbb{R}\}.
\]

\end{lemma}

Next, we analyze the unstable manifold of $p_{\ast}$ in some detail. For this
purpose, it is convenient to introduce a change of variables linked to the
symmetry of $F$ expressed in (A2); namely,%
\[
T(x,y)=(u,v):=\left(  \frac{x+y-2x_{\ast}}{\sqrt{2}},\frac{-x+y}{\sqrt{2}%
}\right)  ,
\]
with inverse%
\[
T^{-1}(u,v)=(x,y):=\left(  x_{\ast}+\frac{u-v}{\sqrt{2}},x_{\ast}+\frac
{u+v}{\sqrt{2}}\right)  .
\]
Clearly, $T$ is a translation of the origin to the fixed point $p_{\ast}$
followed by a counterclockwise rotation of $\pi/4$. The defining map for the
dynamics in the new coordinates can be readily computed to be
\begin{equation}
\tilde{F}(u,v):=T\circ F\circ T^{-1}(u,v)=\left(  \sqrt{2}/R(u,v)\right)
\left(  \sqrt{2}(1+2x_{\ast}^{2})u+(1-x_{\ast}^{2})(v^{2}-u^{2}),-\sqrt
{2}v\right)  ,\label{eq7}%
\end{equation}
where%
\[
R(u,v):=2x_{\ast}(2-x_{\ast})+2\sqrt{2}(1-x_{\ast})u+v^{2}-u^{2}.
\]
The properties of the unstable manifold of the fixed point may now be
described as in the following result, which can be proved directly from Lemma
3.1 and (7).

\begin{lemma}
\label{L2} The unstable manifold of the fixed point $p_{\ast}$, which
corresponds to $0=(0,0)$ in the new $uv$-coordinates, has the form
\[
W^{u}\left(  p_{\ast}\right)  =W^{u}\left(  0\right)  =\{(\varphi
(v),v):\left\vert v\right\vert <1/\sqrt{2}\},
\]
where $\varphi$ is a smooth $(C^{\infty})$ function satisfying the following properties:

\begin{itemize}
\item[(i)] $\varphi(0)=\varphi^{\prime}(0)=0$ and $\varphi(v)\uparrow
\frac{-2+\sqrt{5}}{\sqrt{2}}$ as $v\uparrow1/\sqrt{2}$.

\item[(ii)] $\varphi$ is an even function

\item[(iii)] $\varphi$ satisfies the functional equation
\begin{equation}
\varphi(v)=\kappa\left\{  (1-x_{\ast}^{2})\left(  \varphi(v)^{2}-v^{2}\right)
+\frac{S(\varphi(v),v)}{\sqrt{2}}\varphi\left(  \frac{-2v}{S(\varphi
(v),v)}\right)  \right\}  ,\label{e8}%
\end{equation}
where $\kappa:=1/\left[  \sqrt{2}\left(  1+2x_{\ast}^{2}\right)  \right]  $
and
\[
S(\varphi(v),v):=2x_{\ast}(2-x_{\ast})+2\sqrt{2}(1-x_{\ast})\varphi
(v)+v^{2}-\varphi(v)^{2}.
\]

\end{itemize}
\end{lemma}

\noindent We note here that (8) can be used to obtain Picard iterate (local)
approximations of the unique solution via the recursive formula (\emph{cf}.
\cite{hartman})
\begin{equation}
\varphi_{n+1}(v)=\kappa\left\{  (1-x_{\ast}^{2})\left(  \varphi_{n}%
(v)^{2}-v^{2}\right)  +\frac{S(\varphi_{n}(v),v)}{\sqrt{2}}\varphi_{n}\left(
\frac{-2v}{S(\varphi_{n}(v),v)}\right)  \right\}  .\label{e9}%
\end{equation}
A good starting point for these iterates is%
\[
\varphi_{1}(v):=\sqrt{2}\left(  \sqrt{5}-2\right)  v^{2},
\]
which satisfies the first two properties of Lemma 3.2, and turns out to be a
fairly good approximation for the unstable manifold. Using (9), we obtain an
even better approximation in the form
\begin{equation}
\varphi_{2}(u)=\kappa\left\{  (1-x_{\ast}^{2})\left(  2(\sqrt{5}%
-2)v^{2}-1\right)  v^{2}+\left(  \sqrt{5}-2\right)  S(\varphi_{1}(v),v)\left(
\frac{-2v}{S(\varphi_{1}(v),v)}\right)  ^{2}\right\}  ,\label{e10}%
\end{equation}
where%
\[
S\left(  \varphi_{1}(v),v\right)  =2\left\{  x_{\ast}\left(  1-x_{\ast
}\right)  +v^{2}\left[  \left(  2(1-x_{\ast})(\sqrt{5}-2)+1\right)  -(\sqrt
{5}-2)^{2}v^{2}\right]  \right\}  .
\]
which is illustrated in Fig. \ref{Fig: Picard}. As a matter of fact, it can be
proved that these iterates actually converge locally to the smooth solution of
(8), but the details, which follow the argument in \cite{hartman}, although
straightforward, are a bit too involved to include here. We remark here that
it is not difficult to prove that a global solution of the unstable manifold
equation can obtained as follows:%
\[
W^{u}\left(  p_{\ast}\right)  =\lim{}_{n\rightarrow\infty}\text{ }F^{n}\left(
\Delta^{\prime}\right)  ,
\]
where $\Delta^{\prime}:=\{\left(  x,1-x\right)  :0\leq x\leq1$.

Let us now investigate periodic orbits of the dynamical system (2). It is easy
to see from the definition that $\zeta:=\{(1,0),(0,1)\}$ is a $2$-cycle in
which each point has (least) period two. As for any other cyclic behavior, we
have the following comprehensive result that follows directly from the
definition of the dynamical system and properties (A1)-(A6) and (D1)-(D3).

\begin{lemma}
\label{L3} The $2$-cycle $\zeta:=\{(1,0),(0,1)\}$ is the only cycle of $F$; it
is superstable and has basin of attraction%

\[
\mathfrak{B}(\zeta):=I^{2}\smallsetminus\{(x,x):x\in\mathbb{R}\}.
\]

\end{lemma}

\begin{figure}[ptbh]
\includegraphics[width=.49\textwidth]{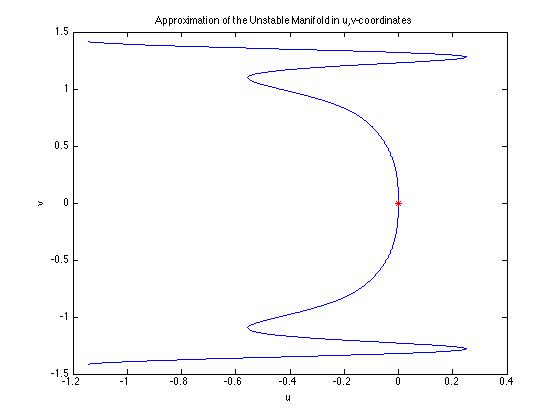}
\includegraphics[width=.49\textwidth]{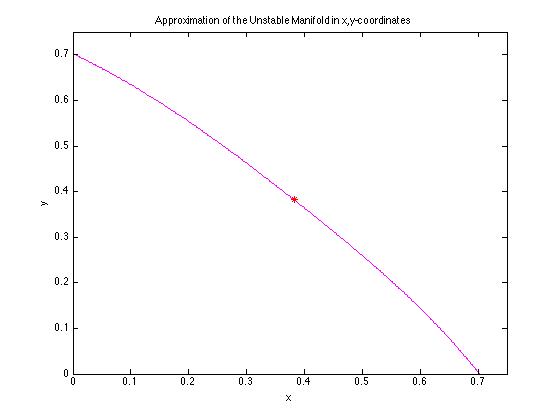}\caption{Picard iteration approximation
($\varphi_{6}$) in $u,v$-coordinates and $x,y$-coordinates respectively.}%
\label{Fig: Picard}%
\end{figure}

\subsection{Summary of the dynamics}

We see from our analysis that the dynamical system (generated by) $F$ defined
by (2)-(3) has highly oscillatory, but quite regular behavior. There is one
fixed point (on the line $y=x$), which attracts everything on the diagonal in
$X$. In addition, there are only two other special points; namely, $(1,0)$ and
$(0,1)$, both of which are superstable periodic points of period two that
comprise the 2-cycle $\zeta:=\{(1,0),(0,1)\}$. Moreover, $\zeta$ attracts
everything in $X$ except the points along the diagonal. Thus we see that our
minimal model nicely replicates the highly oscillatory \textquotedblleft
race\textquotedblright\ behavior found in physical realizations of RSFF
circuits, but none of the chaotic dynamics. In fact, the nonwandering set of
the minimal dynamical system, which actually coincides with the periodic set,
has the very simple form
\begin{equation}
\Omega\left(  F\right)  =Per\left(  F\right)  =\{p_{\ast}\}\cup\zeta
,\label{eq11}%
\end{equation}
and $F$ is $C^{1}$-structurally stable so that the dynamics maintains its
regularity for all perturbations that, along with their derivatives, are
sufficiently small. A proof of the structural stability for our map, which
does not quite satisfy the usual hypotheses, can be fashioned from a
straightforward modification of the methods employed in
\cite{deMelo,palmel,Per,Rob}. Thus, it would seem that perturbations capable
of generating chaos can be very small, but their derivatives need to be quite
large or even fail to exist. We shall verify this in the sequel.

\section{Perturbed Minimal Model with One-dimensional Chaos}

In this section we shall show how small perturbations that do not break the
reflectional symmetry of the minimal model can produce chaos along the
diagonal. More precisely, we shall give a constructive proof of the following result.

\begin{theorem}
\label{T1} There exist continuous arbitrarily small $C^{0}$ perturbations of
the minimal model map $F$ for the RSFF circuit that are symmetric with respect
to reflections in the invariant line $y=x$ and exhibit one-dimensional chaos
in their restrictions to the diagonal.
\end{theorem}

\begin{proof}
The idea of our argument is to embed an arbitrarily
small $C^{0}$ perturbation into $F$ that implants known chaotic dynamics along
the dynamics without breaking the symmetry. This chaotic insert is defined
along the diagonal for any $\sigma>0$ using the continuous piecewise-linear
function
\begin{equation}
\psi_{\sigma}(x):=\left\{
\begin{array}
[c]{cc}%
-x, & -2\sigma\leq x\leq0\\
-2\left(  \sigma-\left\vert x-\sigma\right\vert \right)  , & 0\leq
x\leq2\sigma
\end{array}
\right.  ,\label{e12}%
\end{equation}
which is illustrated in Fig. \ref{Fig: Kink}. Observe that the composition $\tau_{\sigma
}:=\psi_{\sigma}^{2}:[0,2\sigma]\rightarrow\lbrack0,2\sigma]\ $is just the
$2\sigma$-scaled tent map given by%
\[
\tau_{\sigma}(x)=2\left(  \sigma-\left\vert x-\sigma\right\vert \right)  ,
\]
which is known to have chaotic dynamics, including periodic orbits of all
periods, a dense orbit and a Lyapunov exponent of $\log2$ for almost all
initial points (see, e.g. \cite{kathas}).
Now we note from the definitions (2) and (3), (A2) and Lemma 3.1 that the
restriction $f$ of $F$ to the diagonal is%
\begin{equation}
f(x)=\frac{1-x}{2-x},\label{e13}%
\end{equation}
which has the unique (stable) fixed point $x_{\ast}=(1/2)(3-\sqrt{5})$, and
this is a global attractor on the unit interval $[0,1]$. It is easy to see
that for every $\epsilon>0$, we can choose $0<8\sigma<\min\{1/6,\epsilon\}$
such that there is a continuous map $f_{\sigma}:[0,1]\rightarrow\lbrack0,1/2]
$ with the following properties: (i) $f_{\sigma}(x)=$ $\psi_{\sigma}%
(x-x_{\ast})$ for all $x\in J_{\sigma}:=[x_{\ast}-2\sigma,x_{\ast}+2\sigma]$;
(ii) $f_{\sigma}(x)=f(x)$ for all $x\in\lbrack0,1]\smallsetminus\lbrack
x_{\ast}-8\sigma,x_{\ast}+4\sigma]$; (iii) $f_{\sigma}$ is strictly decreasing
on $[0,x_{\ast}+\sigma]\cup\lbrack x_{\ast}+2\sigma,1]$; and (iv) $\left\vert
f_{\sigma}(x)-f(x)\right\vert <4\sigma$ for all $0\leq x\leq1$. We observe
that $J_{\sigma}$ is a global attractor for $f_{\sigma}$, but not strange
since it has dimension equal one.
$\ $The one-dimensional chaotic implant can be extended to the whole
two-dimensional domain of $F$. More precisely, it is easy to see (by simply
linearly joining the perturbed diagonal map to $F$) that there exists for any
$0<8\sigma<\min\{1/6,\epsilon\}$ a continuous map $F_{\sigma}:X\rightarrow X$
satisfying the following properties: (a) the diagonal is $F_{\sigma}$-
invariant (b) $F_{\sigma}$ restricted to the diagonal is $f_{\sigma}$; (c)
$F_{\sigma}$ is symmetric with respect to the diagonal; (d) the dynamics of
$F_{\sigma}$ of the diagonal are qualitatively the same as that of $F$; and
(e) $\left\vert F_{\sigma}(x,y)-F(x,y)\right\vert <8\sigma$ for all $(x,y)\in
X$. As $\sigma$ can be made arbitrarily small, the proof is complete.
\end{proof}

\noindent Observe that the perturbation $F_{\sigma}$ constructed in the above
proof is merely continuous and piecewise smooth inside a small ball centered
at the fixed point $x_{\ast}$ and smooth (i.e.$C^{\infty}$) outside of the
ball. The perturbation and its chaotic dynamics are illustrated in Fig.
\ref{Fig: Chaos1D}. By smoothing the corners of the construction, the
perturbation can be made smooth, which leads directly to the following result.

\begin{corollary}
The perturbation in Theorem \ref{T1} can be chosen so that it is a $C^{\infty
}$ function having the same qualitative dynamics as the function constructed above.
\end{corollary}

It is useful to note that the perturbation $F_{\sigma}$ producing
one-dimensional chaos along the invariant diagonal is arbitrarily $C^{0}$
close to the original map $F$. However, it cannot be made arbitrarily $C^{1}%
$close to $F$ by virtue of the readily verified fact that the restriction to
the diagonal is $C^{1}$-structurally stable, so we cannot find arbitrarily
small perturbations in the $C^{1}$ sense that posses chaotic dynamics. As we
shall see in the next section, there exist arbitrarily small $C^{1}$
perturbations of $F$ that are $C^{\infty}$ and exhibit (two-dimensional)
chaotic dynamical regimes.

\section{Two-dimensional Chaos Induced by Perturbation}

In this section we shall show how arbitrarily small perturbations of the
minimal model can induce chaotic dynamics of various types that are more
substantial and complex than that described in Theorem 4.1. Our first
perturbation involves the (geometric) local embedding of a Smale horseshoe in
the minimal model in a neighborhood of the saddle point $p_{\ast}$.

\subsection{Direct horseshoe chaos}

We begin by defining a map in a neighborhood of $p_{\ast}$ that yields the
desired horseshoe. First, let $0<\delta_{0}\leq0.04$ so that
\[
\bar{B}_{6\sqrt{2}\delta_{0}}(p_{\ast}):=\{(x,y)\in\mathbb{R}^{2}:(x-x_{\ast
})^{2}+(y-y_{\ast})^{2}\leq72\delta_{0}^{2}\}\subset\lbrack x_{\ast}-6\sqrt
{2}\delta_{0},x_{\ast}+6\sqrt{2}\delta_{0}]^{2}\subset(0,1)^{2}=\mathrm{int}%
\left(  I^{2}\right)  ,
\]
where $\mathrm{int}(E)$ denotes the interior of the subset $E$ of the plane.
Now, for each $0<\delta\leq\delta_{0}$,we introduce a $(s,t)$-coordinate
system with origin at $p_{\ast}$, the $s$-axis pointing to the right along the
linear unstable manifold of $p_{\ast}$ and the $t$-axis pointing upward along
the stable manifold of $p_{\ast}$. Let $P_{\delta}:\{(s,t):\left\vert
s\right\vert ,\left\vert t\right\vert \leq6\delta\}\rightarrow\mathbb{R}^{2}$
be defined in terms of $s,t$-coordinates as%
\begin{equation}
P_{\delta}\left(  s,t\right)  :=\left(  \phi_{\delta}(s),\psi_{\delta
}(s,t)\right)  ,\label{e14}%
\end{equation}
where $\phi_{\delta}$ is an odd function such that
\[
\phi_{\delta}(s):=\left\{
\begin{array}
[c]{cc}%
-(3/2)s, & 0\leq s\leq2\delta\\
3(s-3\delta), & 2\delta\leq s\leq4\delta\\
(1/3)\left(  s+5\delta\right)  , & 4\delta\leq s\leq6\delta
\end{array}
\right.  ,
\]
and $\psi_{\delta}$ is an odd function of $t$ for each $s$ given by
\[
\psi_{\delta}(s,t):=-\delta t-\mu(s),
\]
where
\[
\mu(s):=\left\{
\begin{array}
[c]{cc}%
0, & \left\vert s\right\vert \leq2\delta\\
\mathrm{sgn}(s)\left(  \left\vert s\right\vert -\delta\right)  , & \delta\leq
s\leq2\delta\\
\mathrm{sgn}(s)\delta, & \left\vert s\right\vert \geq2\delta
\end{array}
\right.  .
\]

As $p_{\ast}$ is a fixed point of both $P_{\delta}$ and the minimal model $F$,
for and given $\epsilon>0$ we can choose $0<\delta=\delta(\epsilon)\leq
\delta_{0}$ such that in terms of the euclidean norm, we have%
\[
\left\Vert F(p)-P_{\delta}(p)\right\Vert <\epsilon
\]
for all $p\in$ $\bar{B}_{6\sqrt{2}\delta_{0}}(p_{\ast})$. Moreover, it is easy
to see - as shown in Fig. \ref{figA} - that the image of the square
\[
Q_{\delta}:=\{(s,t):-2\delta\leq s,t\leq2\delta\}
\]
under the map (\ref{e12}), namely $P_{\delta}\left(  Q_{\delta}\right)  $, is
a (double) horseshoe. Note that all of the above functions are continuous and
piecewise linear, which means they have smooth approximations (obtained by
smoothly rounding out the corners) that are arbitrarily $C^{0}$-close to them.
Consequently, we can and will assume that the perturbations chosen here to
prove our next result and in the sequel are smooth.

\begin{theorem}
\label{T2} For every $\epsilon>0$ there exists a $0<\delta=\delta
(\epsilon)\leq0.04$ such that the $($smooth$)$ perturbation of the minimal
model map $F$ defined as%
\begin{equation}
F_{\delta}(p):=\left(  1-\rho(r)\right)  P_{\delta}(p)+\rho(r)F(p),\label{e15}%
\end{equation}
where $P_{\delta}$ is as in $($\ref{e14}$)$, $r:=\left\Vert p-p_{\ast
}\right\Vert $ and
\[
\rho(r):=\left\{
\begin{array}
[c]{cc}%
0, & 0\leq r\leq4\delta\sqrt{2}\\
\left(  1/\delta\sqrt{2}\right)  \left(  r-4\delta\sqrt{2}\right)  , &
4\delta\sqrt{2}\leq r\leq5\delta\sqrt{2}\\
1, & r\geq5\delta\sqrt{2}%
\end{array}
\right.  ,
\]
satisfies $\left\Vert F_{\delta}(p)-F(p)\right\Vert <\epsilon$ for all $p\in
I^{2}$ and $F_{\delta}\left(  Q_{\delta}\right)  $ is a double horseshoe as
shown in Fig. \ref{figA}.
\end{theorem}

\begin{proof}
As noted above, the restriction  $0<\delta\leq0.04$ guarantees that $\bar
{B}_{6\sqrt{2}\delta}(p_{\ast})$ is contained in the interior of $I^{2}$ and
that by taking $\delta$ sufficiently small, we can further insure that
$\left\Vert F(p)-P_{\delta}(p)\right\Vert <\epsilon$ on $\bar{B}_{6\sqrt
{2}\delta}(p_{\ast}).$ Hence, it follows from the definition of $\rho$ and (\ref{e15}) that
$\left\Vert F_{\delta}(p)-F(p)\right\Vert <\epsilon$ on $I^{2}$. Finally,
(\ref{e15}) implies that $F_{\delta}\left(  Q_{\delta}\right)  =$ $P_{\delta
}\left(  Q_{\delta}\right)  $, which is the double horseshoe illustrated in
Fig.\ref{figA}, and this completes the proof.
\end{proof}

In light of Theorem \ref{T2}, the next result on the existence of horseshoe
type chaos follows directly from the results of Birkhoff, Moser and Smale
(\emph{cf.} \cite{guho,kathas,Moser2,palmel,Rob,Smale,wigbook}).

\begin{corollary}
The perturbation $F_{\delta}$ in Theorem \ref{T2} is chaotic on an invariant
subset contained in the double horseshoe image described therein. In
particular, $F_{\delta}$ restricted to this invariant set is conjugate to the
shift map on three symbols.
\end{corollary}

\begin{center}
\begin{figure}[th]
\centering
\includegraphics[width=4in]{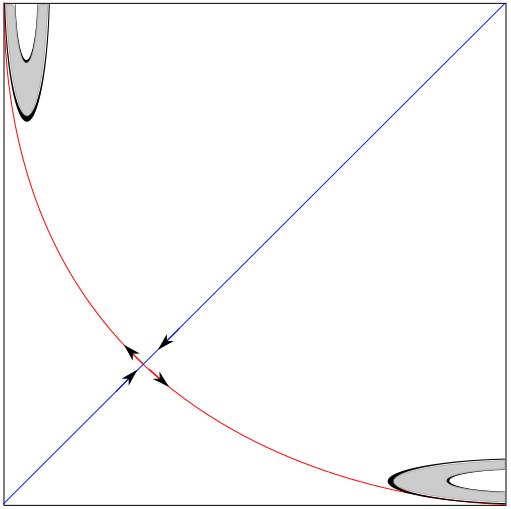}\caption{Embedded double horseshoe}%
\label{figA}%
\end{figure}
\end{center}

\subsection{Snap-back repeller chaos}

There also are arbitrarily small $C^{0}$ perturbations of $F$ exhibiting
snap-back repeller chaos of the type described, for example, by Marotto
\cite{Mar, Mar1}. We start by showing how a small $C^{0}$ perturbation of $F$
can turn $p_{\ast}$ into a source with at least four snap-back points and then
show that it is possible to create infinitely many snap-back points circling
the fixed point. Once again we employ the $s,t$-coordinate system and
$\delta_{0}$ used in the preceding subsection to define the perturbation in a
neighborhood of $p_{\ast}.$ In particular, we define $R_{\delta}%
:\{(s,t):\left\vert s\right\vert ,\left\vert t\right\vert \leq6\delta
\}\rightarrow\mathbb{R}^{2}$ as
\begin{equation}
R_{\delta}\left(  s,t\right)  :=\left(  \phi_{\delta}(s),\psi_{\delta
}(t)\right)  ,\label{e16}%
\end{equation}
where $\phi_{\delta}$ is an odd function of $s$ defined for $s\geq0$ as%
\[
\phi_{\delta}(s):=-2s,
\]
and $\psi_{\delta}$ is an odd function of $t$ defined for $t\geq0$ by
\[
\psi_{\delta}(t):=\left\{
\begin{array}
[c]{cc}%
-2t, & 0\leq t\leq\delta\\
2(t-2\delta), & \delta\leq t\leq3\delta\\
5\delta-t, & 3\delta\leq t\leq6\delta
\end{array}
\right.  ,
\]
and graphed in Fig. \ref{figB}.

As $p_{\ast}$ is a fixed point of both $R_{\delta}$ and the minimal model $F$,
for and given $\epsilon>0$ we can choose $0<\delta=\delta(\epsilon)\leq
\delta_{0}$ such that in terms of the euclidean norm, we have%
\[
\left\Vert F(p)-R_{\delta}(p)\right\Vert <\epsilon
\]
for all $p\in$ $\bar{B}_{6\sqrt{2}\delta_{0}}(p_{\ast})$. \ As above, we are
going to assume with no loss of generality, that our perturbations are
actually smooth, and this leads to our next result, which is an analog of
Theorem \ref{T2}.

\begin{theorem}
\label{T3} For every $\epsilon>0$ there exists a $0<\delta=\delta
(\epsilon)\leq0.04$ such that the $($smooth$)$ perturbation of the minimal
model map $F$ defined as%
\begin{equation}
F_{\delta}(p):=\left(  1-\sigma(r)\right)  R_{\delta}(p)+\sigma
(r)F(p),\label{e17}%
\end{equation}
where $R_{\delta}$ is as in $($\ref{e16}$)$, $r:=\left\Vert p-p_{\ast
}\right\Vert $ and
\[
\sigma(r):=\left\{
\begin{array}
[c]{cc}%
0, & 0\leq r\leq5\delta\sqrt{2}\\
\left(  1/\delta\sqrt{2}\right)  \left(  r-5\delta\sqrt{2}\right)  , &
5\delta\sqrt{2}\leq r\leq6\delta\sqrt{2}\\
1, & r\geq6\delta\sqrt{2}%
\end{array}
\right.  ,
\]
satisfies $\left\Vert F_{\delta}(p)-F(p)\right\Vert <\epsilon$ for all $p\in
I^{2}$ and $p_{\ast}$ is a snap-back repeller for $F_{\delta}$ having the
chaotic dynamics described in \cite{Mar}.
\end{theorem}

\begin{proof}
The restriction $0<\delta\leq0.04$ insures that $\bar{B}_{6\sqrt{2}\delta
}(p_{\ast})$ is contained in the interior of $I^{2}$ and that by taking
$\delta$ sufficiently small, we can further insure that $\left\Vert
F(p)-R_{\delta}(p)\right\Vert <\epsilon$ on $\bar{B}_{6\sqrt{2}\delta}%
(p_{\ast})$. Consequently, (\ref{e17}) and the definition of the function
$\sigma$ implies that $\left\Vert F_{\delta}(p)-F(p)\right\Vert <\epsilon$ on
$I^{2}$. It also follows from (\ref{e17}) that $F_{\delta}(p)=$ $R_{\delta
}(p)$ for all $p\in\bar{B}_{5\sqrt{2}\delta_{0}}(p_{\ast})$, so $p_{\ast}$ is
a hyperbolic repeller for $F_{\delta}$ having - as is clear from Fig. \ref{figB} - four
snap-back points at $(s,t)=(0,\pm2\delta)$ and $(s,t)=(0,\pm5\delta)$. The
chaotic dynamics then follows from \cite{Mar} and the proof is complete.
\end{proof}

\begin{center}
\begin{figure}[th]
\centering
\includegraphics[width=4in]{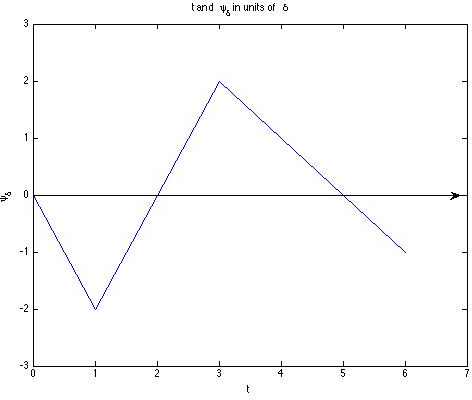}\caption{Coordinate function
for snap-back repeller perturbation}%
\label{figB}%
\end{figure}
\end{center}

It should be noted that the chaos described in Theorem \ref{T3} is essentially
one-dimensional inasmuch as it is confined to the unstable manifold of
$p_{\ast}$ with respect to the minimal model $F$. However, it is not difficult
to see how the construction in the above proof can be modified to obtain
higher dimensional snap-back repeller chaos. One need only consider a
perturbation $\tilde{R}_{\delta}$ given in polar form (with $p_{\ast}$ as the
origin in the $s,t$- coordinate plane) as%
\[
\tilde{R}_{\delta}:=\psi_{\delta}(r)\left(  \cos\theta,\sin\theta\right)  ,
\]
which has two full circles of snap-back points around $p_{\ast}$. If in
addition, we modify $\psi_{\delta}$ in the annulus $2\delta\leq r\leq6\delta$
so that it is negative for certain $\theta$-sectors, the location of the
snap-back points near $p_{\ast}$ can be easily controlled.

\subsection{Chaos generated by embedding transverse homoclinic orbits and
heteroclinic 2-cycles}

Chaos can also be generated by perturbing $F$ so that it has transverse
homoclinic points or transverse heteroclinic 2-cycle points, for then it
exhibits chaotic subshift dynamics (\emph{cf.}%
\cite{ABer,Deng,guho,kathas,palmel,wigbook}).

To begin, we prove a simple result showing how to create transverse
intersections in a homoclinic or heteroclinic curve on a surface with an
arbitrarily small $C^{0}$ perturbation. It should be noted that it is well
known that such transverse intersections can be produced by arbitrarily small
$C^{1}$ perturbations on general $C^{1}$ surfaces (\emph{cf}.
\cite{kathas,palmel,wigbook} ), but we shall find it useful for our
simulations in the sequel to present a specialized $C^{0}$ result for the
plane that is much easier to prove. In fact, the idea of the proof is quite
transparent, involving just a carefully localized small sinusoidal
perturbation normal to the homoclinic or heteroclinic curve (see
Fig.\ref{figC}), but the details are a bit involved. By a \emph{homoclinic
curve} or \emph{heteroclinic curve} joining a single, respectively, pair of
distinct saddle points of a differentiable map of a differentiable surface, we
mean an (open) curve contained in the stable (unstable) and (stable) manifolds
of the single, respectively, pair of points that contains the saddle points in
its closure. A \emph{heteroclinic 2-cycle }is just pair of heteroclinic curves
joining a pair of distinct saddle points.

\begin{lemma}
\label{L4}Let $f:S\rightarrow S$ be a $C^{1}$ self-map of connected surface
$S$ in $\mathbb{R}^{2}$. Suppose that $p$ and $q$ are identical or distinct
saddle points of $f$ joined, respectively, by a homoclinic or heteroclinic
curve $\gamma$ that is contained both in $W^{u}(p)\cap W^{s}(q)$ and a subset
of $S$ that is compact in $\mathbb{R}^{2}$. If there is a point $u_{-1}%
\in\gamma$ such that $u_{0}:=f(u_{-1})$, $u_{1}:=f^{2}(u_{-1})$, $u_{2}:=$
$f^{3}(u_{-1})$ and $u_{3}:=f^{4}(u_{-1})$ are contained in a connected open
neighborhood $U$ of the closed subarc $\kappa$ of $\gamma$ from $u_{-1}$ to
$u_{3}$ and $f$ is invertible on $U$, then there is an arbitrarily small
$C^{0}$ perturbation $g$ of $f$ , equal to $f$ except in an open subset $V$ of
$U$ that contains a nontrivial closed subarc $\sigma$ of $\kappa$ from $u_{0}$
to a point $v_{0}\in\gamma$ not containing $u_{1}$, such that $W_{g}^{u}(p)$
has a transverse intersection with $W_{g}^{s}(q)$ in $f(\sigma)$.
\end{lemma}

\begin{proof}
First, we orient the curve $\gamma$ from $p$ to $q$ so that it follows the
direction of positive iterates. Next, define the (oriented) arclength along
$\gamma$ starting from $u_{0}$ to be $s$, so that $s(u_{0})=0<s(u_{1}%
)<s(u_{2})$ owing to the definition of $\gamma$ and the hypotheses. It follows
from our assumptions that there exists a point $\tilde{u}\in\kappa$ with
$0<s(\tilde{u})<s(u_{1})$ such that the closed subarc $\chi$ of $\gamma$ from
$u_{0}$ to $\tilde{u}$ satisfies the following property:
\begin{itemize}
\item[(P1)] $f^{-1}(\chi)$, $\chi$, $f(\chi)$ and $f^{2}(\chi)$ are pairwise
disjoint closed subarcs of $\ \kappa.$
\end{itemize}
At each point $u\in\gamma$ there is a unit normal vector that is unique if we
specify an orientation, which we do by defining the positive direction to be
consistent with the right hand rule and the positive direction along the
curve. We denote this vector by $\boldsymbol{n}(u)$, which now allows the
definition of a positive and negative distance from $\gamma$ points $x$
sufficiently close to $\gamma$, which we denote by $\nu(x)$. We then combine
this with a coordinate $s(x)$ defined to be $s(\phi(x))$, where $\phi(x):=u$
is the unique (for points sufficiently near $\gamma$) point on the curve such
that the distance $d(x,\gamma)$ from the point $x$ to $\gamma$ is equal to
$\left\Vert x-u\right\Vert $, where $\left\Vert \cdot\right\Vert $ is the
Euclidean norm. Then basic results on normal bundles such as in \cite{Kos}
imply that there exists $\lambda_{0}>0$ such that for every $0<\lambda
\leq\lambda_{0},$%
\begin{equation}
N_{\lambda}:=\left\{  x\in S:s(u_{-1})-\lambda/2<s(x)<s(u_{3})+\lambda
/2;\left\vert \nu(x)\right\vert <\lambda\right\}  \label{eq14}%
\end{equation}
is an open set contained in $U$ and is a neighborhood of the closed subarc
$\kappa$, $N_{\lambda}$ does not contain $p$ or $q$, and it is well defined in
the sense that every $x\in N_{\lambda}$ is uniquely determined by its
coordinates $\left(  s(x),\nu(x)\right)  $. We now have all the tools
necessary to provide a simple description of small $C^{0}$ perturbations that
possess the desired transverse intersections of unstable and stable manifolds.
In virtue of the differentiability of the map and the homoclinic or
heteroclinic curve and the definition of $N_{\lambda}$, there exists
$0<\lambda_{1}<\lambda_{0}$ such that if $0<\lambda\leq\lambda_{1}$ and we
define $\bar{N}_{\lambda}(\chi):=\{x\in S:$ $0\leq s(x)\leq s(\tilde
{u});\left\vert \nu(x)\right\vert \leq\lambda\}$, then the following obtains:
\begin{itemize}
\item[(P2)] $f^{-1}\left(  \bar{N}_{\lambda}(\chi)\right)  $, $\bar
{N}_{\lambda}(\chi)$, $f\left(  \bar{N}_{\lambda}(\chi)\right)  $ and
$f^{2}\left(  \bar{N}_{\lambda}(\chi)\right)  $ are pairwise disjoint closed
subsets of $N_{\lambda_{0}}$.
\end{itemize}
\noindent\noindent Now for any $0<\epsilon<\lambda_{1}$, we define the
perturbation increment function $\Delta_{\epsilon}:S\rightarrow S$ as
\begin{equation}
\Delta_{\epsilon}(x):=\left\{
\begin{array}
[c]{cc}%
\epsilon\left(  1-\lambda^{-1}\left\vert \nu(x)\right\vert \right)
\sin\left(  \frac{2\pi s(x)}{s(\tilde{u})}\right)  \boldsymbol{n}(\phi(x)), &
x\in\bar{N}_{\lambda}(\chi)\\
0, & x\notin\bar{N}_{\lambda}(\chi)
\end{array}
\right.  ,\label{eq15}%
\end{equation}
which is tantamount to saying that the $s$-coordinate function of
$\Delta_{\epsilon}$ is zero and $\nu$-coordinate function is zero except in
$\bar{N}_{\lambda}(\sigma)$ where it is defined to be the coefficient of the
unit normal vector in (\ref{eq15}). Note that this function also vanishes when
$s(x)=0$, $(1/2)s(\tilde{u})$ or $s(\tilde{u})$; that is, all along the
normals to the curve $\gamma$ at $u_{0}$, $\tilde{u}$ and a point $w\in\gamma$
with $s(\tilde{u})<s(w)=(1/2)s(\tilde{u})<s(\tilde{u}).$Moreover, the graph of
(\ref{eq15}) in the $s,\nu$-plane has a transverse intersection with the
$s$-axis at $s=(1/2)s(\tilde{u})$. The desired perturbation of $f$ is%
\begin{equation}
g_{\epsilon}=f+\Delta_{\epsilon},\label{eq16}%
\end{equation}
which is readily seen to have the desired properties. In particular,
$\left\Vert g_{\epsilon}(x)-f(x)\right\Vert \leq\epsilon$ for all $x\in S$,
$W_{g_{\epsilon}}^{u}(p)=W_{f}^{u}(p)$ and $W_{g_{\epsilon}}^{s}(q)=W_{f}%
^{s}(q)$ in a neighborhood of $p$ and $q$, respectively, and the
differentiability of both $f$ and $g_{\epsilon}$ in a neighborhood of $w$
together with (P1), (P2) and the transversality property of $\Delta_{\epsilon
}$ at $w$ guarantee that there is a transverse intersection of $W_{g_{\epsilon
}}^{u}(p)$ and $W_{g_{\epsilon}}^{s}(q)$ at $f(w)$. Thus, the proof is
complete.
\end{proof}

\begin{center}
\begin{figure}[th]
\centering
\includegraphics[width=4in]{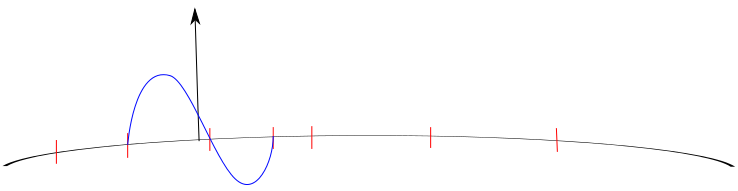}\caption{Transverse
intersection perturbation}%
\label{figC}%
\end{figure}
\end{center}

As a final remark concerning the above lemma, it is rather easy to see how by
smoothing corners and reducing the scale of the perturbation increment, if
necessary, the perturbation of the function can be chosen to be $C^{1}$ small.
Moreover, it is a simple matter to extend the result to general $C^{1} $
surfaces using standard techniques from differential geometry and topology
(\emph{cf.} \cite{kathas,Kos,Rob,palmel,Per,wigbook}).

Next we shall show how to embed $C^{0}$-small dynamics in a neighborhood of
the fixed point $p_{\ast}$ of the minimal model $F$ that has a homoclinic
orbit or a heteroclinic 2-cycle. Once done, we can apply Lemma \ref{L4} to
create chaotic dynamics in arbitrarily small $C^{0}$ perturbations of $F$. We
shall use the time-one maps of the following Hamiltonian differential
equations defined for $\delta>0$ as
\begin{align}
\dot{\xi}  &  =\eta\left(  \delta^{2}+\xi^{2}+\eta^{2}\right)  ,\nonumber\\
\dot{\eta}  &  =\xi\left(  \delta^{2}-\xi^{2}-\eta^{2}\right)  ,\label{e18}%
\end{align}
which has the Hamiltonian function%
\begin{equation}
H_{\delta}=(1/4)\left[  \left(  \xi^{2}+\eta^{2}\right)  ^{2}-2\left(  \xi
^{2}-\eta^{2}\right)  \right]  .\label{e19}%
\end{equation}
Here we have relabeled the $s,t$-coordinates used in subsections 5.1 and 5.2
as $\xi,\eta$-coordinates to avoid confusion with the time parameter $t$ in
(\ref{e18}). This equation yields a pair of homoclinic orbits corresponding to
$H_{\delta}=0$ that are depicted in Fig. \ref{figD}. For the heteroclinic
2-cycle, we choose the system
\begin{align}
\dot{\xi}  &  =\eta,\nonumber\\
\dot{\eta}  &  =-\xi+\left(  \xi^{3}/\delta^{2}\right)  ,\label{e20}%
\end{align}
with Hamiltonian function
\begin{equation}
K_{\delta}=\left(  1/4\right)  \left[  2\left(  \xi^{2}+\eta^{2}\right)
-\left(  x^{4}/\delta^{2}\right)  \right]  .\label{e21}%
\end{equation}
This system has a heteroclinic 2-cycle contained in $K_{\delta}=\delta^{2}/4$,
which is also shown in Fig. \ref{figD}. Now let the solution of (\ref{e18})
and (\ref{e20}) be denoted, respectively as
\begin{equation}
\left(  \xi,\eta\right)  =\varphi_{\delta}\left(  t,\left(  \xi_{0},\eta
_{0}\right)  \right) \label{eq22}%
\end{equation}
and%
\begin{equation}
\left(  \xi,\eta\right)  =\psi_{\delta}\left(  t,\left(  \xi_{0},\eta
_{0}\right)  \right)  .\label{e23}%
\end{equation}

The desired embeddings are obtained from the time-1 maps of the above; namely,
we define%
\begin{equation}
\Phi_{\delta}(\xi,\eta):=\varphi_{\delta}\left(  1,\left(  \xi,\eta\right)
\right) \label{e24}%
\end{equation}
and%
\begin{equation}
\Psi_{\delta}(\xi,\eta):=\psi_{\delta}\left(  1,\left(  \xi,\eta\right)
\right)  .\label{e25}%
\end{equation}
Clearly, these maps have the desired homoclinic and heteroclinic orbits,
respectively. Now, we select $\delta_{0}>0$ so small that $\bar{B}%
_{6\delta_{0}}(p_{\ast})$ is contained in the interior of $I^{2}$ and then
$0<\delta=\delta(\epsilon)\leq\delta_{0}$ for a given $\epsilon>0$ such that
\begin{equation}
\left\Vert F(z)-\Phi_{\delta}(z)\right\Vert <\epsilon\label{e26}%
\end{equation}
and%
\begin{equation}
\left\Vert F(z)-\Psi_{\delta}(z)\right\Vert <\epsilon\label{e27}%
\end{equation}
for all $z\in\bar{B}_{6\delta_{0}}(p_{\ast})$. Whence, we can define
\begin{equation}
\hat{F}_{\delta}(z):=\left(  1-\varkappa(r)\right)  \Phi_{\delta}%
(z)+\varkappa(r)F(z)\label{eq28}%
\end{equation}
and%
\begin{equation}
\tilde{F}_{\delta}(z):=\left(  1-\varkappa(r)\right)  \Psi_{\delta
}(z)+\varkappa(r)F(z),\label{eq29}%
\end{equation}
where $r:=\left\Vert z-p_{\ast}\right\Vert $ and%
\begin{equation}
\varkappa(r):=\left\{
\begin{array}
[c]{cc}%
0, & 0\leq r\leq3\delta\\
\left(  1/\delta\right)  \left(  r-3\delta\right)  , & 3\delta\leq
r\leq4\delta\\
1, & r\geq4\delta
\end{array}
\right.  .\label{e30}%
\end{equation}

\begin{center}
\begin{figure}[th]
\centering
\includegraphics[width=0.49\textwidth]{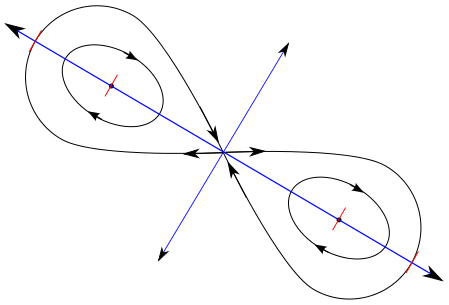}\includegraphics[width=0.49\textwidth]{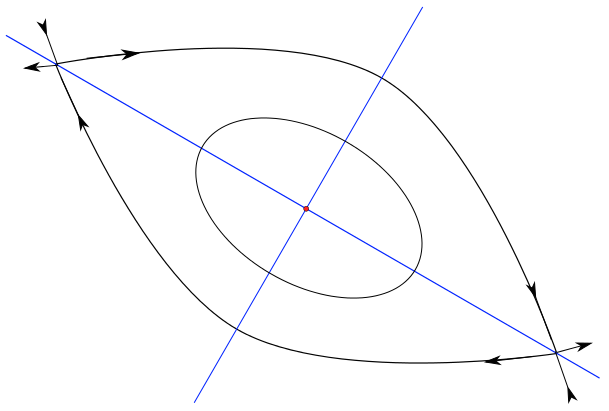}\caption{Embedding
homoclinic and heteroclinic orbits}%
\label{figD}%
\end{figure}
\end{center}

It follows from (\ref{e20})-(\ref{eq29}) that $\left\Vert F(z)-\hat{F}%
_{\delta}(z)\right\Vert <\epsilon$ and $\left\Vert F(z)-\tilde{F}_{\delta
}(z)\right\Vert <\epsilon$ for all $z\in I^{2}$ and that $\hat{F}_{\delta}$
has a homoclinic orbit comprised of the stable and unstable manifolds at
$p_{\ast}$ and $\tilde{F}_{\delta}$ has a heteroclinic 2-cycle centered at
$p_{\ast}$ as shown in Fig. \ref{figD}. Whence, we may, after choosing
$\delta$ smaller if necessary, use Lemma \ref{L4} to further perturb $\hat
{F}_{\delta}$ and $\tilde{F}_{\delta}$ to create a transverse intersection
point in the homoclinic orbit and, respectively, to create a transverse
intersection point in one or both components of the heteroclinic 2-cycle,
while still remaining $\epsilon$-close to the original maps in the $C^{0}$
metric. Then $\lambda$-Lemma based arguments of the type developed in such
sources as \cite{ABer,Deng,palmel,wigbook} and the smoothing discussed above
lead directly to the following result.

\begin{theorem}
\label{T4} For every $\epsilon>0$ there exist smooth $C^{0}$ $\epsilon$-close
perturbations of the minimal model map $F$ exhibiting transverse homoclinic
point or transverse heteroclinic 2-cycle induced chaos.
\end{theorem}

\subsection{Multihorseshoe strange attractor perturbations}

The next element of our description and analysis of more pervasive chaotic
$C^{0}$ perturbations of the map $F$ defined by (\ref{eq2})-(\ref{eq3})
involves embedding a symmetric pair of attracting horseshoes along the lines
introduced in \cite{JB}. In particular, we shall show how to $C^{0}$ perturb
$F$ to produce multihorseshoe (more precisely, double-horseshoe) chaos
(\emph{cf.}\cite{JB}, as shown in Fig. \ref{figE}.

To begin the construction of the embedding, we note first that the minimal map
$F$ is actually defined and smooth at all points in the $x,y$-plane above the
curve defined by $x+y-xy=0$. Next, for convenience, we denote the period-2
points $(1,0)$ and $(0,1)$ by $p$ and $q$, respectively. Let $\epsilon>0$ be
given. As both $p$ and $q$ are fixed points of $F^{2}$ and of $G^{2}$, where
$G$ is the reflection in the line $y=x$, we may choose $0<\delta\leq1/4$ such
that
\begin{equation}
\left\Vert F(z)-G(z)\right\Vert <\epsilon/2\label{e5.4.1}%
\end{equation}
whenever $z\in\bar{B}_{10\sqrt{2}\delta}(p)\cup\bar{B}_{10\sqrt{2}\delta}(q)$.
Define
\begin{equation}
\tilde{F}_{\epsilon}(z):=\left\{
\begin{array}
[c]{cc}%
F(z), & z\notin\bar{B}_{10\sqrt{2}\delta}(p)\cup\bar{B}_{10\sqrt{2}\delta
}(q)\\
\left(  1-\omega(r_{p})\right)  G(z)+\omega(r_{p})F(z), & z\in\bar{B}%
_{10\sqrt{2}\delta}(p)\\
\left(  1-\omega(r_{q})\right)  G(z)+\omega(r_{q})F(z), & z\in\bar{B}%
_{10\sqrt{2}\delta}(q)
\end{array}
\right.  ,\label{e5.4.2}%
\end{equation}
where $r_{p}:=\left\Vert z-p\right\Vert $, $r_{q}:=\left\Vert z-q\right\Vert $
and
\begin{equation}
\omega(r):=\left\{
\begin{array}
[c]{cc}%
0, & 0\leq r\leq8\delta\sqrt{2}\\
\left(  1/\delta\sqrt{2}\right)  \left(  r-8\delta\sqrt{2}\right)  , &
8\delta\sqrt{2}\leq r\leq9\delta\sqrt{2}\\
1, & r\geq9\delta\sqrt{2}%
\end{array}
\right.  .\label{e5.4.3}%
\end{equation}
It is clear from the definition and (\ref{e5.4.1}) that
\begin{equation}
\left\Vert F(z)-\tilde{F}_{\epsilon}(z)\right\Vert <\epsilon/2.\label{e5.4.4}%
\end{equation}

It now remains to make a final adjustment of the map $\tilde{F}_{\epsilon}$
that has the attracting horseshoes. To this end, we define the following
perturbation of the identity in $\bar{B}_{10\sqrt{2}\delta}(p)\cap I^{2}$
using the original $x,y$-coordinates%
\[
\Theta\left(  x,y\right)  =\left(  \varphi(x),\psi(x,y)\right)  ,
\]
where
\begin{equation}
\varphi(x):=\left\{
\begin{array}
[c]{cc}%
1, & x\leq1-7\delta\text{ or }1-\delta\leq x\leq1\\
2(x+\delta)-1, & 1-4\delta\leq x\leq1-\delta\\
-2(x+7\delta)+3, & 1-7\delta\leq x\leq1-4\delta
\end{array}
\right. \label{e5.4.5}%
\end{equation}
and%
\begin{equation}
\psi(x,y):=\left\{
\begin{array}
[c]{cc}%
y/5, & x\leq1-9\delta\text{ or }1-5\delta\leq x\leq1\\
\left(  y/5\right)  +6\left(  1-5\delta-x\right)  , & 1-6\delta\leq
x\leq1-5\delta\\
\left(  y/5\right)  +6\delta & 1-7\delta\leq x\leq1-6\delta\\
\left(  y/5\right)  +3\left(  x+9\delta-1\right)  & 1-9\delta\leq
x\leq1-7\delta
\end{array}
\right.  .\label{e5.4.6}%
\end{equation}
Observe that $\Theta$ has a sink at $p$ and a saddle point at $\left(
1-2\delta,0\right)  $ with a horizontal unstable and vertical stable manifold.
Moreover, $\Theta$ maps the rectangle $[1-4\delta,1-\delta]\times
\lbrack-\delta,9\delta]$ onto a piecewise smooth (and of course smoothable)
attracting horseshoe as defined in \cite{JB}. This function can be reflected
in the line $y=x$ to obtain the symmetric map $\hat{\Theta}:=G\circ\Theta\circ
G$ mapping the rectangle $[-\delta,9\delta]\times\lbrack1-4\delta,1-\delta]$
onto the reflection of the horseshoe image of $\Theta$ in $y=x.$

Now it follows directly from the above definitions that, taking $\delta$
smaller if necessary, we can insure that
\begin{equation}
\left\Vert G(z)-\Theta\circ G(z)\right\Vert <\epsilon/2\label{e5.4.7}%
\end{equation}
when $z\in\bar{B}_{10\sqrt{2}\delta}(q)\cap I^{2}$ and
\begin{equation}
\left\Vert G(z)-\hat{\Theta}\circ G(z)\right\Vert <\epsilon/2\label{e5.4.8}%
\end{equation}
for all $z\in\bar{B}_{10\sqrt{2}\delta}(p)\cap I^{2}$. Therefore, also taking
into account (\ref{e5.4.4}), the modification $F_{\epsilon}$ of $\tilde
{F}_{\epsilon}$ defined as
\[
F_{\epsilon}(z):=\left\{
\begin{array}
[c]{cc}%
F(z), & z\notin\bar{B}_{10\sqrt{2}\delta}(p)\cup\bar{B}_{10\sqrt{2}\delta
}(q)\\
\left(  1-\omega(r_{p})\right)  \hat{\Theta}\circ G(z)+\omega(r_{p})F(z), &
z\in\bar{B}_{10\sqrt{2}\delta}(p)\\
\left(  1-\omega(r_{q})\right)  \Theta\circ G(z)+\omega(r_{q})F(z), & z\in
\bar{B}_{10\sqrt{2}\delta}(q)
\end{array}
\right.
\]
satisfies
\[
\left\Vert F(z)-F_{\epsilon}(z)\right\Vert <\epsilon.
\]
Moreover, we note that both $p$ and $q$ are attracting fixed points (sinks) of
$F_{\epsilon}^{2}$, just as they are of $F^{2}$, while $(1-2\delta,0)$ and
$(0,1-2\delta)$ comprise a 2-cycle of $F_{\epsilon}$ (consisting of saddle
points of $F_{\epsilon}^{2}$) but not $F$. Finally, as constructed, each of
the symmetric horseshoes shown in Fig. E is an attracting horseshoe of $F$, so
that it follows from the main multihorseshoe theorem in \cite{JB} that we have
now proved the following result.

\begin{theorem}
\label{T5} For every positive $\epsilon$ there is a smooth $C^{0}$ $\epsilon$
- close perturbation $F_{\epsilon}$ of the minimal map $F$ having a strange
chaotic double-horseshoe attractor with attracting horseshoes in neighborhoods
of the points $(1,0)$ and $(0,1)$.
\end{theorem}

It is interesting to note that the strange attractor for the small
perturbation of $F$ in Theorem \ref{T5} resembles a discrete analog of the
\textquotedblleft double scroll\textquotedblright\ attractor for Chua's
circuit (\emph{cf}. \cite{chua,CWHZ}), also seen in the dynamics of RSFF
realization simulations such as in \cite{CG} and could, with some minor
modification, produce discrete analogs of the attractors found in the various
physical circuit (ODE) models such as in \cite{KA,kac,LMH,msd,okt}.

\begin{center}
\begin{figure}[th]
\centering
\includegraphics[width=4in]{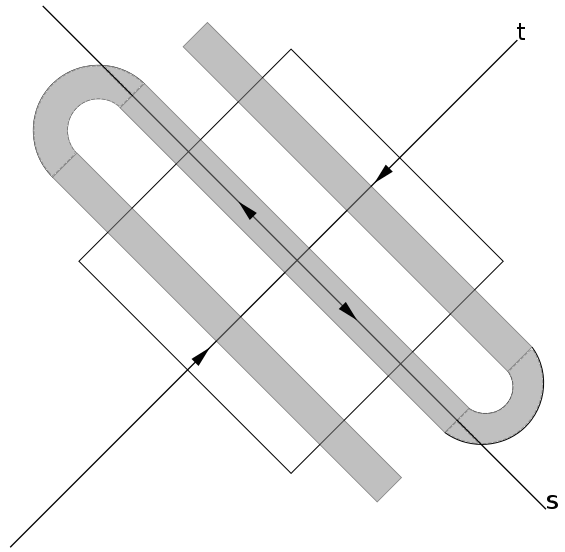}\caption{Embedded double-horseshoe
chaotic strange attractor}%
\label{figE}%
\end{figure}
\end{center}

\subsection{Neimark--Sacker bifurcation perturbations}

Our final result shall show that Neimark--Sacker bifurcations can occur in
arbitrarily small $C^{0}$ perturbations of the minimal model map $F$. Once
again, we use the $\xi,\eta$-coordinates employed in preceding subsections to
define for $\delta>0$ the following parameter-dependent map in polar
coordinates in the $\xi,\eta$ - plane with origin at $p_{\ast}$%
\begin{equation}
N_{\delta}\left(  \xi,\eta;\mu,\alpha\right)  :=-\delta\tanh(\mu
r/\delta)\left(  \cos(\theta+\alpha),\sin(\theta+\alpha)\right)
,\label{e5.5.1}%
\end{equation}
where $1/2<\mu<3/2$, $\alpha$ is nonnegative and $r:=\left\Vert (\xi
,\eta)-p_{\ast}\right\Vert $. It is easy to verify that the origin is a global
attractor of (\ref{e5.5.1}) for $1/2<\mu<1$ and a local repeller for
$1<\mu<3/2$, so $\mu=1$ is a bifurcation value. Moreover, as $\mu$ increases
across $1$ with a corresponding transition of the origin from a sink to a
source, a stable invariant circle of radius $r=r(\mu)$, where $r(\mu)$ is the
unique positive solution of%
\begin{equation}
\delta\tanh(\mu r/\delta)=r.\label{e5.5.2}%
\end{equation}
The parameter $\alpha$ just represents the rotation of the map (\ref{e5.5.1}),
so what we have is a Neimark--Sacker bifurcation at $\mu=1$.

Now just as in the preceding subsections, we can choose $\delta_{0}>0$ so
small that $\bar{B}_{6\delta_{0}}(p_{\ast})$ is contained in the interior of
$I^{2}$ and then $0<\delta=\delta(\epsilon)\leq\delta_{0}$ for a given
$\epsilon>0$ such that
\begin{equation}
\left\Vert F(z)-N_{\delta}(z)\right\Vert <\epsilon\label{e5.5.3}%
\end{equation}
for all $\mu\in(1/2,3/2)$ and $\alpha\geq0$ whenever $z\in\bar{B}_{6\delta
_{0}}(p_{\ast}).$ Therefore, by defining
\[
\check{F}_{\delta}(z):=\left(  1-\varkappa(r)\right)  N_{\delta}%
(z)+\varkappa(r)F(z),
\]
where $\varkappa$ is defined just as in (\ref{e30}), we obtain a map that is
$\epsilon$-close to $F$ in the $C^{0}$ norm and has the desired bifurcation
properties. In short, we have now proved the following result.

\begin{theorem}
\label{T6} For every positive $\epsilon$ there is a smooth $C^{0}$ $\epsilon$
- close perturbation $\check{F}$ of the minimal map $F$ having a
Niemark--Sacker bifurcation at $p_{\ast}$. \
\end{theorem}

A direct detailed construction was used for the proof of Theorem \ref{T6}, but
the same result can be proved, using the main theorem in \cite{CB}, for any
perturbation in which $p_{\ast}$ changes from an attractor to a local repeller
as a parameter is varied. It should also be noted that one can, by a
straightforward modification of the above procedure, construct arbitrarily
small $C^{0}$ perturbations of $F$ exhibiting doubling Neimark--Sacker (Hopf)
cascades like those in the ad hoc RSFF model analyzed in \cite{BRS}.

\section{Simulations, Computations and Comparisons}

Our purpose in this section is to show with just a few examples that our
discrete dynamical model - when properly perturbed - shares many properties
with actual physical realizations (and their associated mathematical models)
of the \emph{R}-\emph{S} flip-flop and related circuits such as in
\cite{Danca, KA, kac, LMH, Moser, msd, okt, rzw, ztkbn}. For example, we have
already demonstrated in Section 4 that our model can be perturbed so that it
exhibits the chaos found by simulation in a one-dimensional map associated
with the dynamics of the realization of the flip-flop circuit - which is not
an R-S flip-flop circuit - investigated in \cite{okt}, and it is this kind of
chaos we consider in the next subsection.

\subsection{Perturbed minimal model with one-dimensional chaos}

We derive the perturbed map by plugging in (\ref{e12}) to (\ref{eq3}) where
$x\rightarrow\phi_{\sigma}(x)$ and $y\rightarrow\phi_{\delta}(y)$,%

\begin{equation}
\label{Eq: RSFF}%
\begin{split}
\xi(x,y) := \frac{\phi_{\delta}(y)(1-\phi_{\sigma}(x))}{1-(1-\phi_{\sigma
}(x))(1-\phi_{\delta}(y))}\\
\eta(x,y) := \frac{\phi_{\sigma}(x)(1-\phi_{\delta}(y))}{1-(1-\phi_{\sigma
}(x))(1-\phi_{\delta}(y))}.
\end{split}
\end{equation}

The iterates of the perturbed map are shown in Fig. \ref{Fig: Chaos1D}.

\begin{figure}[ptbh]
\begin{subfigure}{.49\linewidth}
\includegraphics[height = 1.9in]{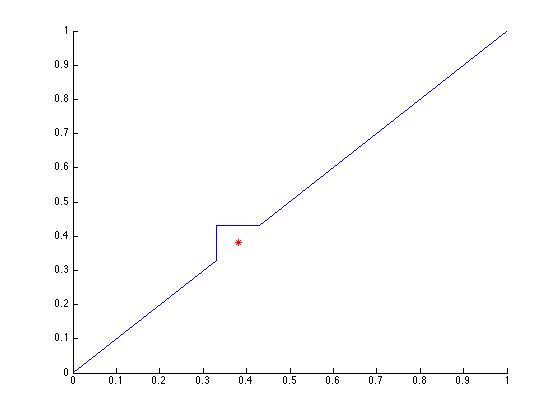}
\subcaption{Illustration of Eq. \ref{e12}.\label{Fig: Kink}}
\end{subfigure}
\begin{subfigure}{.49\linewidth}
\includegraphics[height = 1.9in]{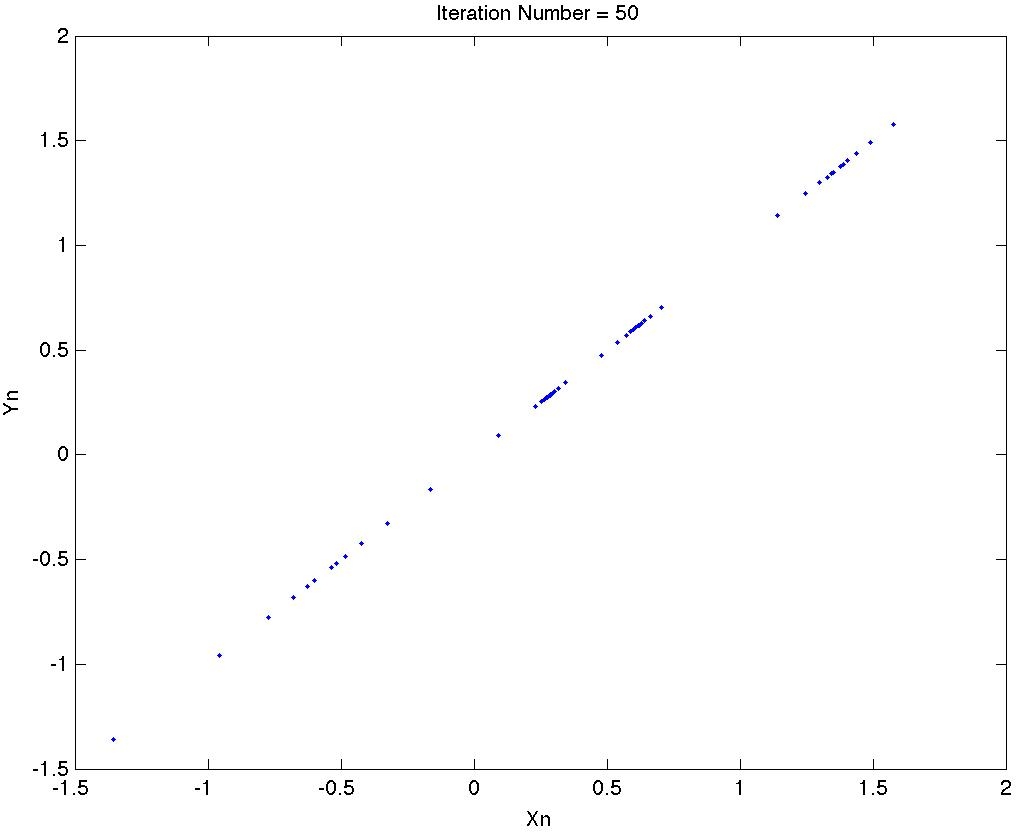}
\subcaption{Plot of the iterates of the perturbed map.\label{Fig: Chaos1D}}
\end{subfigure}
\end{figure}

\begin{corollary}
The perturbation in the preceding theorem can be chosen so that it is a
$C^{\infty}$ function having the same qualitative dynamics as the function
constructed above.
\end{corollary}

\subsection{Evidence of two-dimensional chaos}

As we have shown above, almost any type of chaotic dynamics - including double
scroll chaos - can be obtained from the ideal map by inserting specific
localized perturbations. We shall now show that the introduction of a fairly
general type of small $C^{0}$ perturbation is apt to produce chaotic dynamics.
In order to produce two-dimensional chaos lets use the perturbation,
\begin{align*}
\label{Eq: 2Dpert}\varphi_{k,\sigma}(x)  &  =3.7[a_{0}-.02-\sum_{j=1}^{k}%
a_{j}\text{cos}j{\pi}x]\\
a_{0,\sigma}(x)  &  =\frac{\mu_{x}+\lambda_{x}}{8}+\frac{\frac{1}{2}+\epsilon
}{4}\\
a_{j,\sigma}(x)  &  =\frac{1}{2j\pi}\{(1-\lambda_{x})\text{sin}j(\frac{1}%
{2}-\epsilon)\pi+\frac{\mu_{x}+\lambda_{x}}{j\pi}[(-1)^{j}-1]\}
\end{align*}
to get the map,
\begin{equation}
\label{Eq: 2Dpert}%
\begin{split}
\xi &  :=\frac{\varphi_{m,\delta}(y)(1-\varphi_{n,\epsilon}(x))}%
{1-(1-\varphi_{m,\delta}(y))(1-\varphi_{n,\epsilon}(x))}\\
\eta &  :=\frac{\varphi_{n,\epsilon}(x)(1-\varphi_{m,\delta}(y))}%
{1-(1-\varphi_{m,\delta}(y))(1-\varphi_{n,\epsilon}(x))}%
\end{split}
\end{equation}

The iterates of this map, with their characteristic splattering indicative of
chaos, are shown in Fig. \ref{Fig: Chaos2D}.

\begin{figure}[ptbh]
\begin{center}
\includegraphics[scale=.26]{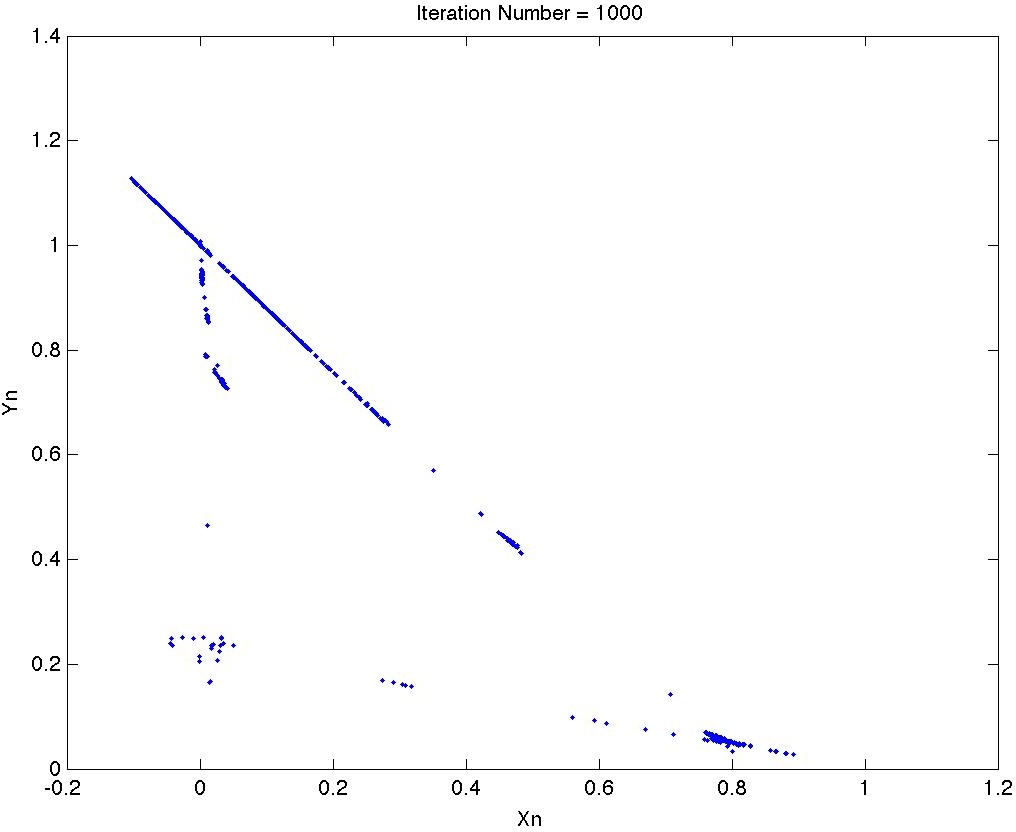}
\end{center}
\caption{Iterates of a perturbed map seeming to exhibit two-dimensional
chaos.}%
\label{Fig: Chaos2D}%
\end{figure}

\subsection{Ring oscillator example}

Another interesting circuit, which we intend to analyze in detail in a
forthcoming paper, is a modified ring oscillator. This ring oscillator
comprises three NOR gates with feedbacks, in a very similar fashion to that of
the RS flip-flop circuit. One may even choose to think of this as a
three-dimensional RS flip-flop circuit.

\begin{figure}[ptbh]
\begin{center}
\includegraphics[scale=.75]{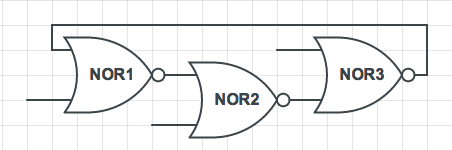}
\end{center}
\caption{A schematic of a ring oscillator designed out of three NOR gates.}%
\end{figure}

We note that if the inputs are set to zero, the system is precisely a ring
oscillator because the NOR gates now act as inverters. Applying our algorithm
for finding a discrete dynamical system model for logical circuits yields the
following \textquotedblleft ideal\textquotedblright\ (unperturbed) model,%

\begin{align}
\label{Eq: Ring_Ideal}\xi= \frac{(x-1)(y(z-1)+1)}{(x-1(y-1)(z-1)-1)}%
,\nonumber\\
\eta= \frac{(y-1)(z(x-1)+1)}{(x-1(y-1)(z-1)-1)},\\
\zeta= \frac{(z-1)(x(y-1)+1)}{(x-1(y-1)(z-1)-1)};\nonumber
\end{align}

The only valid fixed point is $(x_{\ast},y_{\ast},z_{\ast})=((3-\sqrt
{5})/2,(3-\sqrt{5})/2,(3-\sqrt{5})/2)$. Since the fixed points are roots of a
quartic equation, one may wonder what happens to the other roots. We find that
the other three roots are out of our domain, which means that they are of no
practical consequence. For the ideal model, regardless of the initial
conditions, the orbits decay to the fixed point in a spiral manner shown in
Fig. \ref{Fig: Ring_Ideal}.

\begin{figure}[ptbh]
\includegraphics[scale=.14]{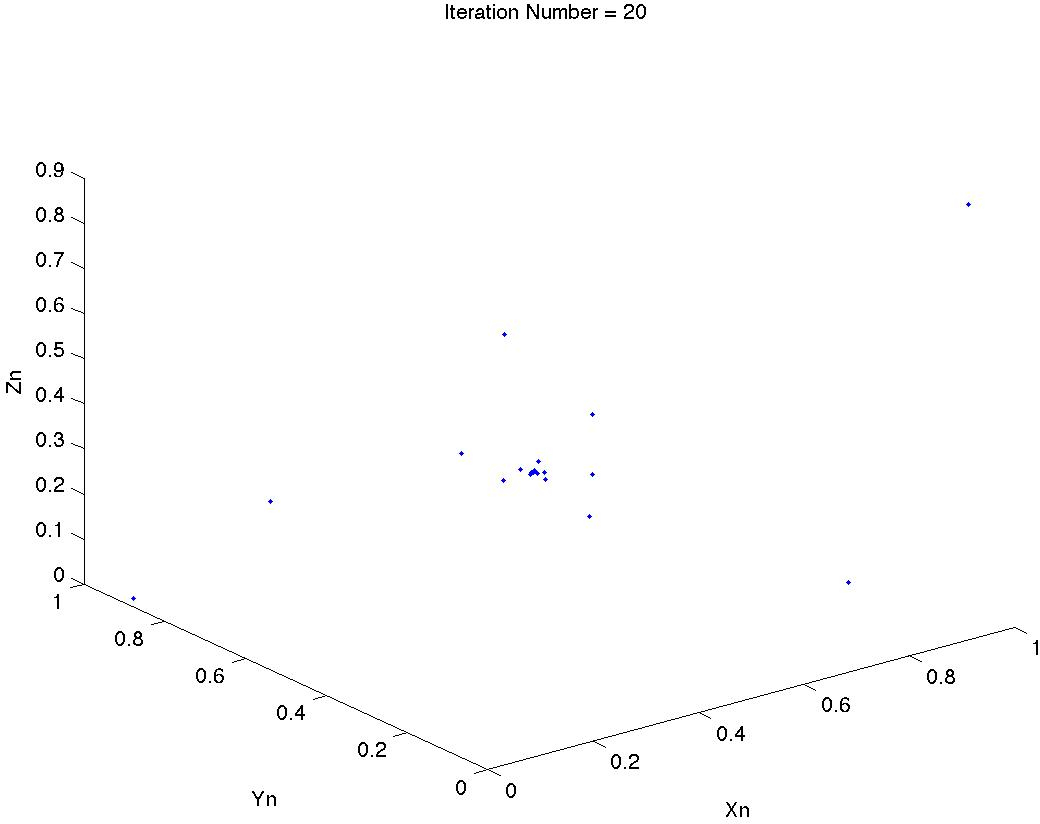}
\includegraphics[scale=.14]{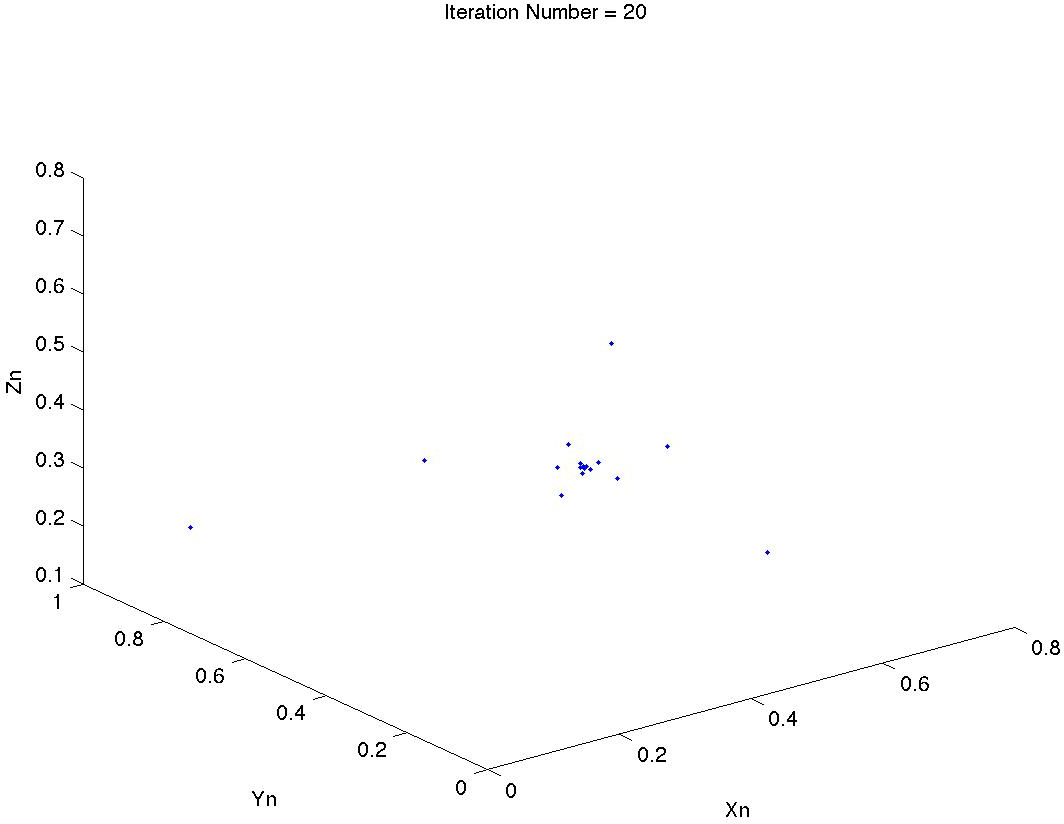}
\includegraphics[scale=.14]{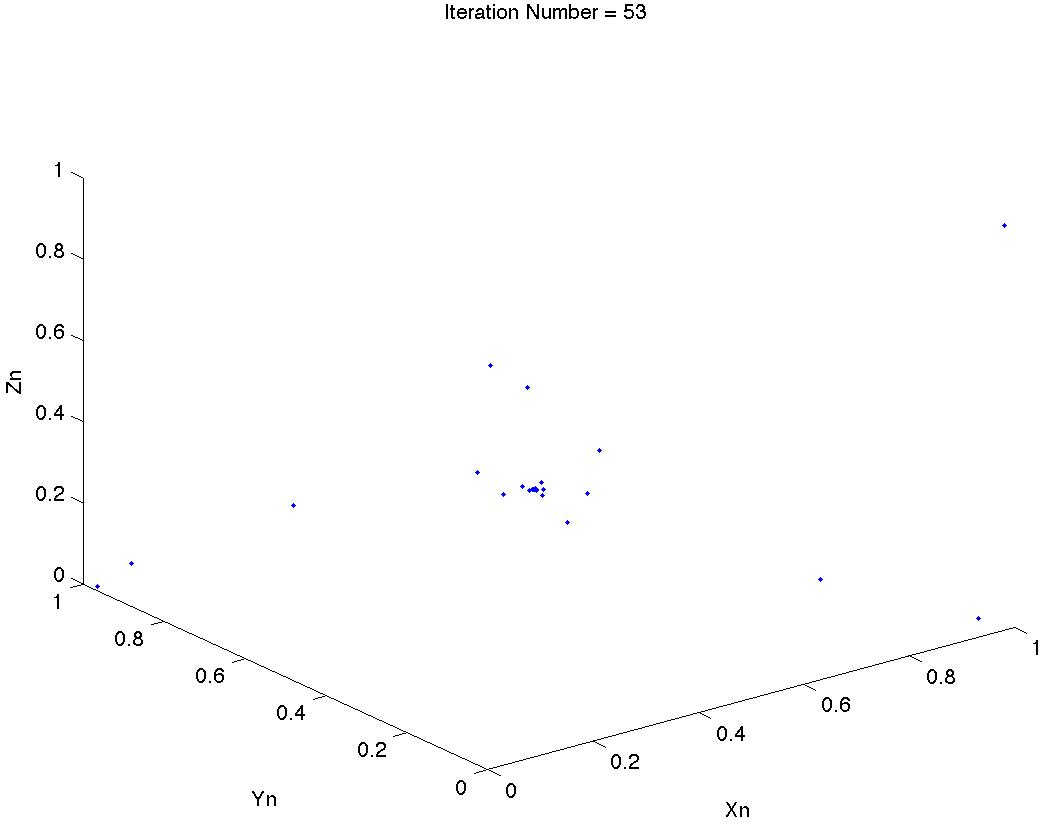}\caption{Plots of the orbits of
Eq. \ref{Eq: Ring_Ideal} with initial conditions, $(.99,.1,.9)$, $(.1,.9,.2)$,
and $(.1,.1,.9)$ respectively.}%
\label{Fig: Ring_Ideal}%
\end{figure}

In order to produce more interesting dynamics we perturb the ideal model
slightly. Using a perturbation similar to (\ref{Eq: 2Dpert}), we produce the
chaotic dynamics shown in Fig. \ref{Ring_Chaotic}.

\begin{figure}[ptbh]
\begin{center}
\includegraphics[scale=.65]{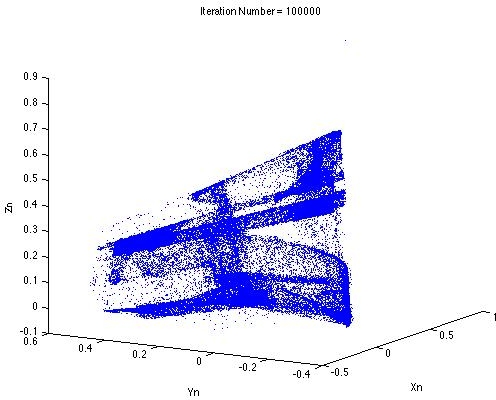}
\end{center}
\caption{Plots of the orbits of the perturbed model with initial conditions
$(.99,.1,.9)$.}%
\label{Ring_Chaotic}%
\end{figure}

\section{Concluding Remarks}

We have introduced and analyzed a rather simple discrete (ideal) dynamical
model - grounded on first principles - for the RSFF circuit, which is based
upon the iterates of a planar map. Moreover, we have proved that this model
can be modified - by arbitrarily small $C^{0}$ perturbations - to produce just
about any dynamical property observed in physical realizations of RSFF and
related flip-flop circuits. We have also shown that rather general small
perturbations are apt to change the structurally stable ideal model to a
two-dimensional discrete dynamical system with chaotic dynamics and related
artifacts such as strange chaotic attractors.

Naturally, we are planning to extend this discrete dynamical systems approach
to a much wider class of logical circuits and their perturbations, and verify
the effectiveness of this approach by comparing our dynamic predictions with
those of more standard ODE approaches as well as experimental data extracted
from actual physical circuit measurements. We also plan to address a number of
related questions such as showing that logical circuit realizations represent
perturbations of our ideal models in some sense, and that the perturbations
can actually be characterized and quantified. Such an investigation should
provide insight into the relationship that we believe exists between our
discrete dynamical systems approach and the underlying discrete dynamics of
reconfigurable chaotic logic gates \cite{DMM}.

\section*{Acknowledgements}

The authors would like to thank Ian Jordan for sharing his expertise on
logical circuits, which proved to be very helpful in this investigation.

\end{document}